\documentclass[prl,superscriptaddress,twocolumn,nofootinbib]{revtex4}
\usepackage{amsthm}
\usepackage{amsmath,latexsym,amssymb,verbatim,enumerate,graphicx}
\usepackage{color}
\usepackage{amsfonts,amssymb,amsmath}
\usepackage[]{graphics,graphicx,epsfig}
\usepackage{amsthm}
\usepackage{mathtools}
\newtheorem{definition}{Definition} 

\newtheorem{lemma}[definition]{Lemma}

\newtheorem{thm}[definition]{Theorem}

\newtheorem*{rep@theorem}{\rep@title}
\newcommand{\newreptheorem}[2]{%
\newenvironment{rep#1}[1]{%
 \def\rep@title{#2 \ref{##1} (restatement)}%
 \begin{rep@theorem}}%
 {\end{rep@theorem}}}
\makeatother

\newreptheorem{thm}{Theorem}
\newreptheorem{lem}{Lemma}

\def\ba#1\ea{\begin{align}#1\end{align}}
\def\ban#1\ean{\begin{align*}#1\end{align*}}

\newcommand{\be}{\begin{equation}}
\newcommand{\ee}{\end{equation}}

\def\Eveoutput{\mathbb{Z}}

\def\benum{\begin{enumerate}}
\def\eenum{\end{enumerate}}

\def\squareforqed{\hbox{\rlap{$\sqcap$}$\sqcup$}}
\def\qed{\ifmmode\squareforqed\else{\unskip\nobreak\hfil
\penalty50\hskip1em\null\nobreak\hfil\squareforqed
\parfillskip=0pt\finalhyphendemerits=0\endgraf}\fi}
\def\endenv{\ifmmode\;\else{\unskip\nobreak\hfil
\penalty50\hskip1em\null\nobreak\hfil\;
\parfillskip=0pt\finalhyphendemerits=0\endgraf}\fi}

\def\Eveoutput{z}
\def\dcintro{d_c}

\usepackage{mathtools}
\newcommand{\defeq}{\vcentcolon=}

\newcommand{\<}{\langle}
\renewcommand{\>}{\rangle}

\def\be{\begin{equation}}
\def\ee{\end{equation}}
\def\ben{\begin{eqnarray}}
\def\een{\end{eqnarray}}

\def\bei{\begin{itemize}}
\def\eei{\end{itemize}}

\mathchardef\ordinarycolon\mathcode`\:
\mathcode`\:=\string"8000
\def\vcentcolon{\mathrel{\mathop\ordinarycolon}}
\begingroup \catcode`\:=\active
  \lowercase{\endgroup
  \let :\vcentcolon
  }

\newcommand{\nc}{\newcommand}
 \nc{\proj}[1]{|#1\rangle\!\langle #1 |} 
\nc{\avg}[1]{\langle#1\rangle}

\nc{\conv}{\operatorname{conv}}
\nc{\smfrac}[2]{\mbox{$\frac{#1}{#2}$}} \nc{\Tr}{\operatorname{Tr}}
\nc{\ox}{\otimes} \nc{\dg}{\dagger} \nc{\dn}{\downarrow}
\nc{\lmax}{\lambda_{\text{max}}}
\nc{\lmin}{\lambda_{\text{min}}}

\nc{\csupp}{{\operatorname{csupp}}}
\nc{\qsupp}{{\operatorname{qsupp}}} \nc{\var}{\operatorname{var}}
\nc{\rar}{\rightarrow} \nc{\lrar}{\longrightarrow}
\nc{\poly}{\operatorname{poly}}
\nc{\polylog}{\operatorname{polylog}} \nc{\Lip}{\operatorname{Lip}}
\nc{\Om}{\Omega}
\nc{\wt}[1]{\widetilde{#1}}

\def\>{\rangle}
\def\<{\langle}

\def\bu{{\it \textbf{u}}}
\def\bx{\textbf{x}}

\def\pxutot {P(\bx_1,\ldots,\bx_n|\bu_1,\ldots,\bu_n,\Eveoutput, e)}
\def\xuseq{\Eveoutput,e, \bx_1,\bu_1,\ldots,\bx_n,\bu_n}

\nc{\glneq}{{\raisebox{0.6ex}{$>$}  \hspace*{-1.8ex} \raisebox{-0.6ex}{$<$}}}
\nc{\gleq}{{\raisebox{0.6ex}{$\geq$}\hspace*{-1.8ex} \raisebox{-0.6ex}{$\leq$}}}

\nc{\vholder}[1]{\rule{0pt}{#1}}
\nc{\wh}[1]{\widehat{#1}}
\nc{\h}[1]{\widehat{#1}}

\nc{\ob}[1]{#1}

\def\beq{\begin {equation}}
\def\eeq{\end {equation}}

\def\be{\begin{equation}}
\def\ee{\end{equation}}

\def\ACC{\text{ACC}}

\def\ACCd{\text{ACC}}

\def\Lazuma{L}

\def\Lazuma{L}
\nc{\eq}[1]{(\ref{eq:#1})} 
\nc{\eqs}[2]{\eq{#1} and \eq{#2}}

\nc{\eqn}[1]{Eq.~(\ref{eqn:#1})}
\nc{\eqns}[2]{Eqs.~(\ref{eqn:#1}) and (\ref{eqn:#2})}

\nc{\region}{\cS\cW}

\newenvironment{protocol*}[1]
  {
    \begin{center}
      \hrulefill\\
      \textbf{#1}
  }
  {
    \vspace{-1\baselineskip}
    \hrulefill
    \end{center}
  }

\begin{document}


\title{Randomness amplification under minimal fundamental assumptions on the devices}

\author{Ravishankar Ramanathan}
\affiliation{Institute of Theoretical Physics and Astrophysics, National Quantum Information Centre, Faculty of Mathematics, Physics and Informatics, University of Gda\'nsk, 80-308 Gda\'nsk, Poland}

\author{Fernando G.S.L. Brand\~{a}o}
\affiliation{Quantum Architectures and Computation Group, Microsoft Research, Redmond, Washington 98052, USA}
\affiliation{Department of Computer Science, University College London, WC1E 6BT London, UK}

\author{Karol Horodecki}
\affiliation{Institute of Informatics, National Quantum Information Centre, Faculty of Mathematics, Physics and Informatics, University of Gda\'nsk, 80-308 Gda\'nsk, Poland}

\author{Micha{\l} Horodecki}
\affiliation{Institute of Theoretical Physics and Astrophysics, National Quantum Information Centre, Faculty of Mathematics, Physics and Informatics, University of Gda\'nsk, 80-308 Gda\'nsk, Poland}

\author{Pawe{\l} Horodecki}
\affiliation{Faculty of Applied Physics and Mathematics, National Quantum Information Center, Gda\'nsk University of Technology, 80-233 Gda\'nsk, Poland}

\author{Hanna Wojew\'{o}dka}
\affiliation{Institute of Theoretical Physics and Astrophysics, National Quantum Information Centre, Faculty of Mathematics, Physics and Informatics, University of Gda\'nsk, 80-308 Gda\'nsk, Poland}
\affiliation{Institute of Mathematics, Faculty of Mathematics, Physics and Chemistry, University of Silesia, Bankowa 14, 40-007 Katowice,
Poland}

\date{\today}

\begin{abstract}
Recently, the physically realistic protocol amplifying the randomness of Santha-Vazirani sources producing cryptographically secure random bits 
was proposed; however for reasons of practical relevance, the crucial question remained open whether this can be accomplished under the minimal conditions necessary for the task. Namely, is it possible to achieve randomness amplification using only two no-signaling components and in a situation where the violation of a Bell inequality only guarantees that some outcomes of the device for specific inputs exhibit randomness?
Here, we solve this question and present a device-independent protocol for randomness amplification of Santha-Vazirani sources using a device consisting of two non-signaling components. We show that the protocol can amplify any such source that is not fully deterministic into a fully random source while tolerating a constant noise rate and prove the composable security of the protocol against general no-signaling adversaries. Our main innovation is the proof that even the partial randomness certified by the two-party Bell test (a single input-output pair ($\textbf{u}^*, \textbf{x}^*$) for which the conditional probability $P(\textbf{x}^* | \textbf{u}^*)$ is bounded away from $1$ for all no-signaling strategies that optimally violate the Bell inequality) can be used for amplification. We introduce the methodology of a partial tomographic procedure on the empirical statistics obtained in the Bell test that ensures that 
the outputs constitute a linear min-entropy source of randomness. As a technical novelty that may be of independent interest, we prove that the Santha-Vazirani source satisfies an exponential concentration property given by a recently discovered generalized Chernoff bound. 
\end{abstract}

\maketitle

\textit{Introduction.-}
Random number generators are ubiquitous, finding applications in varied domains such as statistical sampling, computer simulations and gambling scenarios. Certain physical phenomena such as radioactive decay or thermal radiation have high natural entropy, there are also computational algorithms that produce sequences of apparently random bits. In many cryptographic tasks however, it is necessary to have trustworthy sources of randomness. As such, developing device-independent protocols for generating random bits is of paramount importance. 

We consider the task of randomness amplification, to convert a source of partially random bits to one of fully random bits. The paradigmatic model of a source of randomness is the Santha-Vazirani (SV) source \cite{SV}, a model of a biased coin where the individual coin tosses are not independent but rather the bits $Y_i$ produced by the source obey
\begin{equation}
\label{SVdef}
\frac{1}{2} - \varepsilon \leq P(Y_i = 0 | Y_{i-1}, \dots, Y_1) \leq \frac{1}{2} + \varepsilon.
\end{equation}
Here $0 \leq \varepsilon < \frac{1}{2}$ is a parameter describing the reliability of the source, the task being to convert a source with $\varepsilon < \frac{1}{2}$ into one with $\varepsilon \rightarrow 0$. Interestingly, this task is known to be impossible with classical resources, a single SV source cannot be amplified \cite{SV}.

In \cite{Renner}, the non-local correlations of quantum mechanics were shown to provide an advantage in the task of amplifying an SV source. A device-independent protocol for generating truly random bits was demonstrated starting from a critical value of $\varepsilon (\approx 0.06)$ \cite{Renner, Grudka}, where device-independence refers to the fact that one need not trust the internal workings of the device. An improvement was made in \cite{Acin} where using an arbitrarily large number of spatially separated devices, it was shown that one could amplify randomness starting from any initial $\varepsilon < \frac{1}{2}$. In \cite{our2}, we demonstrated a device-independent protocol which uses a \textit{constant} number of spatially separated components and amplifies sources of arbitrary initial $\varepsilon < \frac{1}{2}$ while simultaneously tolerating a constant amount of noise in its implementation. All of these protocols were shown to be secure against general adversaries restricted only by the no-signaling principle of relativity under a technical assumption of independence between the source and the device. In \cite{CSW14}, a randomness amplification protocol was formulated for general min-entropy sources and shown to be secure against quantum adversaries without the independence assumption, the drawback of this protocol being that it requires a device with a large number of spatially separated components for its implementation. Other protocols have also been proposed \cite{Pawlowski, Plesch}, for which full security proofs are missing. For fundamental as well as practical reasons, it is vitally important to minimize the number of spatially separated components in the protocol. As such, devising a protocol with the minimum possible number of components (two space-like separated ones for a protocol based on a Bell test) while at the same time, allowing for robustness to errors in its implementation is crucial. 

Let $\textbf{U}, \textbf{X}$ denote the input and output sets respectively, of honest parties in a device-independent Bell-based protocol for randomness amplification. 
A necessary condition for obtaining randomness against general no-signaling (NS) attacks is that 
for some input $\textbf{u}^* \in \textbf{U}$, output $\textbf{x}^* \in \textbf{X}$ and a constant $c < 1$, 
\textit{every} no-signaling box $\{P(\textbf{x} | \textbf{u})\}$ that obtains the observed Bell violation has $P(\textbf{x} = \textbf{x}^* | \textbf{u} = \textbf{u}^*) \leq c$. i.e., 
\ben
\label{min-req}
\exists (\textbf{x}^*, \textbf{u}^*) \; \; \text{s.t.} \; \; &&\forall \{P(\textbf{x} | \textbf{u})\} \; \; \text{with} \; \; \textbf{B} \cdot \{P(\textbf{x} | \textbf{u})\} = 0 \nonumber \\
&&P(\textbf{x} = \textbf{x}^* | \textbf{u} = \textbf{u}^*) \leq c < 1,
\een
where $\textbf{B}$ is an indicator vector (with entries $B(\textbf{x}, \textbf{u})$) encoding the Bell expression and $\textbf{B} \cdot \{P(\textbf{x} | \textbf{u})\} = \sum_{\textbf{x}, \textbf{u}} B(\textbf{x}, \textbf{u}) P(\textbf{x} | \textbf{u}) = 0$ denotes that the box $\{P(\textbf{x} | \textbf{u})\}$ algebraically violates the inequality. 
Note that while the Bell inequality violation guarantees Eq.(\ref{min-req}) for some $\textbf{x}^*, \textbf{u}^*$ for each NS box, here the requirement is for a strictly bounded \textit{common} entry $P(\textbf{x} = \textbf{x}^* | \textbf{u} = \textbf{u}^*)$ for all boxes leading to the observed Bell violation. It is straightforward to see that if Eq. (\ref{min-req}) is not met, then the observed Bell violation does not guarantee any randomness and a device-independent protocol to obtain randomness cannot be built on the basis of this violation. 
If in addition to the necessary condition in Eq. (\ref{min-req}), we also had for the same input-output pair $(\textbf{u}^*, \textbf{x}^*)$ that 
\be
\label{eq:suff-cond} 
\tilde{c} \leq P(\textbf{x} = \textbf{x}^* | \textbf{u} = \textbf{u}^*) 
\ee
for some constant $\tilde{c} > 0$, then clearly all the outputs for input $\textbf{u}^*$ possess randomness and extraction of this randomness may be feasible.  

Here, we present a fully device-independent protocol that allows to amplify the randomness of any $\varepsilon$-SV source under the minimal necessary condition in Eq. (\ref{min-req}). A novel element of the protocol is an additional test (to the usual Bell test) akin to partial tomography of the boxes that the honest parties perform, to lower bound (in a linear number of runs) $P(\textbf{x} = \textbf{x}^* | \textbf{u} = \textbf{u}^*) =: \textbf{D} \cdot \{P(\textbf{x} | \textbf{u}) \}$. Here $\textbf{D}$ is an indicator vector with entries $D(\textbf{x}, \textbf{u})$ such that $D(\textbf{x}, \textbf{u}) = 1$ iff $(\textbf{x}, \textbf{u}) = (\textbf{x}^*, \textbf{u}^*)$. This test ensures that additionally Eq.(\ref{eq:suff-cond}) is also met for a sufficient number of runs, a detailed description is provided in the Supplemental Material. The protocol uses a device consisting of only two no-signaling components and tolerates a constant error rate. We show that the output bits from the protocol satisfy universally-composable security, the strongest form of cryptographic security, for any adversary limited only by the no-signaling principle. 

\textit{Main Result.-}
We present a two-party protocol to amplify the randomness of SV sources against no-signaling adversaries, formally we show the following (the detailed security proof is presented in the Supplemental Material):

\begin{thm} [informal]  \label{mainthm}
For every $\varepsilon < \frac{1}{2}$, there is a protocol using an $\varepsilon$-SV source and a device consisting of two no-signaling components with the following properties:

\begin{itemize}
\item Using the device $\poly(n, \log(1/\gamma))$ times, the protocol either aborts or produces $n$ bits which are $\gamma$-close to uniform and independent of any no-signaling side information about the device and classical side information about the source (e.g. held by an adversary).

\item Local measurements on many copies of a two-party entangled state, with $\poly(1 - 2\varepsilon)$ error rate, give rise to a device that does not abort the protocol with probability larger than $1 - 2^{- \Omega(n)}$.

\end{itemize}
The protocol is non-explicit and runs in $\poly(n, \log{(1/\gamma)})$ time. {Alternatively it can use an explicit extractor to produce a single bit of randomness that is $\gamma$-close to uniform in $\poly(\log{(1/\gamma)})$ time.}
\end{thm}

%
%
%


\begin{figure}
\begin{protocol*}{Protocol I}
\begin{enumerate}
\item The $\varepsilon$-SV source is used to choose the measurement settings $u = (\bu^1_{\leq n}, \bu^2_{\leq n})$ for $n$ runs on the single device consisting of two components. The device produces output bits $x = (\bx^1_{\leq n}, \bx^2_{\leq n})$.
\item The parties perform an estimation of the violation of the Bell inequality in the  device by computing the empirical average $\Lazuma_n(x,u) \defeq  \frac{1}{n} \sum_{i=1}^{n} B(\textbf{x}_i, \textbf{u}_i)$. The protocol is aborted unless $\Lazuma_n(x,u) \leq \delta$ for fixed constant $\delta > 0$. 
\item Conditioned on not aborting in the previous step, the parties subsequently check if $\textit{S}_n(x,u) \defeq \frac{1}{n} \sum_{i=1}^{n} D(\textbf{x}_i, \textbf{u}_i) \geq \mu_1$. The protocol is aborted if this condition is not met for fixed $\mu_1 > 0$.  
\item Conditioned on not aborting in the previous steps, the parties apply an independent source extractor \cite{CG, one_bit_extr} to the sequence of outputs from the device and further $n$ bits from the SV source.
\end{enumerate}
\end{protocol*}
\caption{Protocol for device-independent randomness amplification from a single device with two no-signaling components.}
\label{protocolsingle}
\end{figure}

\textit{Protocol.-}
The protocol for the task of randomness amplification from $\varepsilon$-SV sources is given explicitly in Fig. \ref{protocolsingle} and illustrated in Fig. \ref{fig:bip-fig-prot}, its structure is as follows. The two honest parties Alice and Bob use bits from the $\varepsilon$-SV source to choose the inputs to their no-signaling boxes in multiple runs of a Bell test and obtain their respective outputs. They check for the violation of a Bell inequality and abort the protocol if the test condition is not met. The novel part of the protocol is a subsequent test that the honest parties perform which ensures when passed that a sufficient number of runs were performed with 
boxes that have randomness in their outputs. If both tests are passed, the parties apply a randomness extractor to the output bits and some further bits taken from the SV source. The output bits of the extractor constitute the output of the protocol, which we show to be close to being fully random and uncorrelated from any no-signaling adversary. 

\begin{figure}[t]
\includegraphics[scale=0.4]{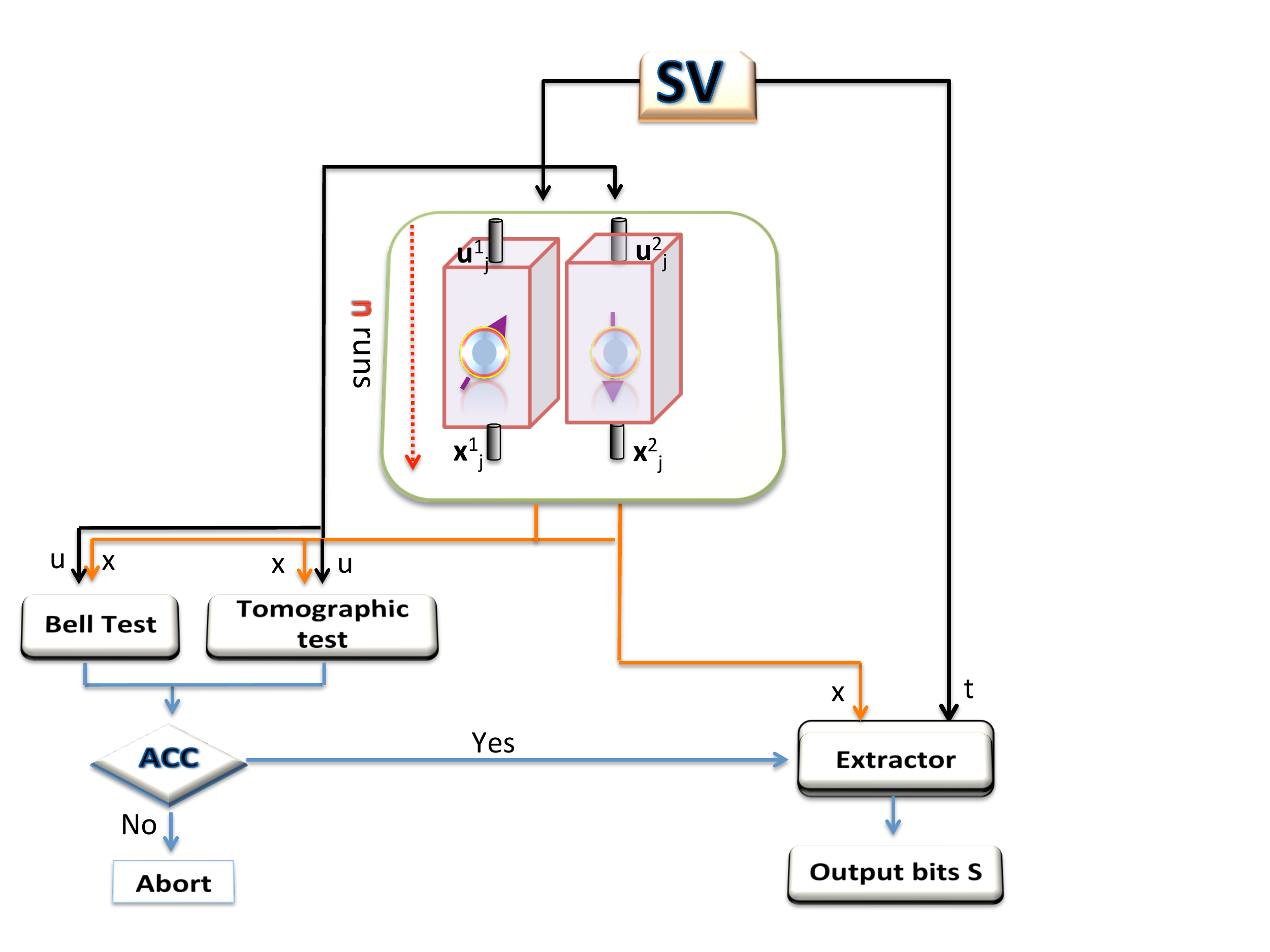}
\caption{An illustration of the protocol for randomness amplification using two no-signaling components. The bits from the SV source (black arrows) are used as inputs $(\textbf{u}^1_j, \textbf{u}^2_j)$ for the $j$-th run of the two spatially separated devices, with $1 \leq j \leq n$, and the corresponding outputs $(\textbf{x}^1_j, \textbf{x}^2_j)$ are obtained. The inputs and outputs of all the $n$ runs $(u,x)$ are subjected to two tests: a Bell test for the violation of a specific Bell inequality and a (partial) tomographic test counting a specific number of input-output pairs $(\textbf{u}^*, \textbf{x}^*)$. If both tests are passed (denoted by $\ACC$), the outputs $x$ (orange arrows) are hashed together with further $n$ bits $t$ from the SV source using an extractor.}
\label{fig:bip-fig-prot}
\end{figure}

\textit{Description of the setup.-}
The setup of the protocol is as follows. The honest parties and Eve share a no-signaling box $\{p(x, z|u', w)\}$ where $u' = \textbf{u'}_{\leq n} := (\textbf{u'}_{1}, \dots, \textbf{u'}_n)$ and $x = \textbf{x}_{\leq n} := (\textbf{x}_1, \dots, \textbf{x}_n)$ denote the input and output,  respectively, of the honest parties for the $n$ runs of the protocol, with $w$ and $z$ being the inputs and outputs of the adversary Eve. The devices held by the honest parties are separated into $2$ components with corresponding inputs and outputs $u'^i$ and $x^i$, respectively, for $i = 1, 2$, i.e., $u' = (u'^1, u'^2)$ and $x = (x^1, x^2)$. Note that $u'^i, x^i$ themselves denote the inputs and outputs of the $n$ runs of the protocol for party $i$, i.e., $u'^i = \textbf{u'}^i_{\leq n} := (\textbf{u'}^i_1, \dots, \textbf{u'}^i_n)$ and $x^i = \bx^i_{\leq n} = (\bx^i_1, \dots, \bx^i_n)$. 
Here, for the $j$-th run of the Bell test, we have labeled the measurement settings of Alice $\textbf{u'}^1_j$ and those of Bob $\textbf{u'}^2_j$ with the corresponding outcomes $\textbf{x}^1_j$ and $\textbf{x}^2_j$, and denoted the joint inputs of Alice and Bob in this run as $\textbf{u'}_j = (\textbf{u'}^1_j, \textbf{u'}^2_j)$ with corresponding joint output $\textbf{x}_j = (\textbf{x}^1_j, \textbf{x}^2_j)$. 
The honest parties draw bits $u$ from the SV source to input into the box, i.e., they set $u' = u$. They also draw further $n$ bits $t$, which will be fed along with the outputs $x$ into the randomness extractor to obtain the output of the protocol $s := \text{Ext}(x,t)$. The adversary has classical information $e$ correlated to $u, t$. The box we consider for the protocol is therefore given by the family of probability distributions $\{p(x,z,u,t,e|u',w)\}$.  

\textit{Assumptions.-}
Let us first state formally the assumptions on $\{p(x,z,u,t,e|u',w)\}$, see also \cite{our2}. 
\begin{itemize}
\item {\bf No-signaling (shielding) assumption:} 
The box satisfies the constraint of no-signaling between the honest parties and Eve as well as between the different components of the device
\ben 
\label{eq:ns1}
p(x|u', w) &=& p(x|u'), \nonumber \\
p(z|u',w) &=& p(z|w), \nonumber \\
p(x^i|u') &=& p(x^i|u'^i) \; \; \; i = 1, 2.
\label{eq:as-nosig}
\een

Each device component also obeys a time-ordered no-signaling (\texttt{tons}) condition for the $k \in [n]$ runs performed on it:
\begin{eqnarray}
\label{eq:tons1}
&&p(x^i_{k}|z,u'^i, w, u, t, e) = \nonumber \\
&& \; \; \;   p(x^i_{k} | z,u'^i_{\leq k},w,u, t, e) \; \; \; \forall k \in [n]
\label{eq:tons}
\end{eqnarray}
where $u'^i_{\leq k} := u'^i_{1}, \ldots, u'^i_{k}$.

\item {\bf SV conditions:} 
The variables $(u,t,e)$ form an SV source, that is satisfy 
Eq. (\ref{SVdef}). 
In particular,
$p(t|u,e)$ and $p(u|e)$ also obey the SV source conditions. The fact that $e$ cannot be perfectly correlated to $u$, $t$ is called the \textit{private SV source assumption} \cite{our2}.

\item {\bf Assumption A1:} The devices do not signal to the SV source, i.e., the distribution of $(u,t,e)$ is 
independent of the inputs $(u',w)$:
\begin{eqnarray}
&&\sum_{x,z} p(x,z,u,t,e|u',w) = p(u,t,e) \; \; \; \forall{(u, t, e, u', w)}.
\label{eq:assumption1}
\end{eqnarray}

\item {\bf Assumption A2:} The box is fixed independently of the SV source:
\begin{eqnarray}
&&p(x,z|u',w,u,t,e) = p(x,z|u',w) \; \;  \forall{(x, z, u',w, u, t,e)}. \nonumber \\
\label{eq:assumption2}
\end{eqnarray}
\end{itemize}
In words, the main assumptions are that the different components of the device do not signal to each other and to the adversary Eve. Additionally, there is also a time-ordered no-signaling ($tons$) structure assumed on different runs of a single component, the outputs in any run may depend on the previous inputs within the component but not on future inputs. Moreover, we also assume that the structure of the box $p(x,z|u',w)$ is fixed independently of the SV source $p(u,t,e)$, i.e., the box is an unknown and arbitrary input-output channel independent of the SV source. This precludes malicious correlations such as in the scenario where for each bit string $u$ taken from the source, a different (possibly local) box tuned to $u$ is supplied, in which case the Bell test may be faked by local boxes \cite{our3}. Finally, it is worth noting that no randomness may be extracted under the assumptions stated above in a classical setting, whereas the Bell violation by quantum boxes allows to amplify randomness in a device-independent setting. 

\textit{Security definition.-}
For $\Lazuma_n(x,u) = \frac{1}{n} \sum_{i=1}^{n} B(\textbf{x}_i, \textbf{u}_i)$, the first (Bell) test in the protocol is passed when $\Lazuma_n(x,u) \leq \delta$. We define the set $\ACC_1$ as the set of $(x,u)$ such that this test 
is passed:
\be 
\ACC_1 \defeq \{(x,u) : \Lazuma_n(x,u) \leq \delta\}. 
\ee
The $\delta$ is the noise parameter in the Bell test which is chosen to be a positive constant depending on the initial $\varepsilon$ of the SV source, going to zero in the limit of $\varepsilon \rightarrow \frac{1}{2}$ 
(see Theorem 8 in the Supplemental Material). 
Similarly, we define $\ACC_2$ as the set of $(x,u)$ for which the second test is passed, i.e., 
\be 
\ACC_2 \defeq \left\lbrace (x,u) : \textit{S}_n(x,u) \geq \mu_1 \right\rbrace.
\ee
We also define the set $\ACC = \ACC_1 \cap \ACC_2$ of $(x,u)$ for which both tests in the protocol are passed and $\ACC_u$ as the cut
\be 
\ACC_u \defeq \{ x : (x,u) \in \ACC \}.
\ee 

After $u$ is input as $u'$ and conditioned on the acceptance of the tests $\ACC$, applying the independent source extractor \cite{CG, Xin-Li, one_bit_extr} $s=\text{Ext}(x,t)$, one gets the box
\begin{eqnarray}
&&p(s,z,e|w,\ACC ) \nonumber \\
&&\; \; \equiv  \sum_{u}\sum_{\text{Ext}(x,t)=s}   p(x,z,u,t,e|w,\ACC ). 
\end{eqnarray}

The composable security criterion is now defined in terms of the distance 
of $p(s,z,e|w, \ACC )$ to an ideal box $p^{id} = \frac{1}{|S|} p(z,e|w,\ACC)$ with $p(z,e|w,\ACC) = \sum_{s} p(s,z,e|w,\ACC)$. Formally, the security is given by the distance $d_c$ defined as 
\be
\label{dist-comp}
d_c :=\sum_{s,e} \max_{w } \sum_{z} \left |p(s,z,e|w, \ACC ) - \frac{1}{|S|}p(z,e|w, \ACC )\right|.
\ee


\textit{Outline of the proof.-}
The proof of security of the protocol is a modification of the proof we presented in \cite{our2} with the crucial differences being due to the weak randomness that the two-party Bell inequality violation gives and an additional partial tomographic test imposed on the device. 

To amplify SV sources, one needs Bell inequalities where quantum theory can achieve the maximal no-signaling value of the inequality \cite{Renner}, failing which, for sufficiently small $\varepsilon$, the observed correlations may be faked with classical deterministic boxes. However, Bell inequalities with this property are not sufficient, this is exemplified by the tripartite Mermin inequality \cite{Mermin, Renner}. This inequality is algebraically violated in quantum theory using a GHZ state, however for any function of the measurement outcomes one can find no-signaling boxes which achieve its maximum violation and for which this particular function is deterministic thereby providing an attack for Eve to predict with certainty the final output bit. While \cite{Acin} and \cite{our2} considered Bell inequalities with more parties, the problem of finding two-party algebraically violated Bell inequalities (known as pseudo-telepathy games) \cite{RMP-Bell14} with the property of randomness for some function of the measurement outcomes was open. Unfortunately, none of the bipartite Bell inequalities tested so far have the property that \textit{all} no-signaling boxes which maximally violate the inequality have randomness in any function of the measurement outcomes $f(\textbf{x})$ for some input $\textbf{u}$ in the sense that for all such boxes 
\be
\label{eq:strong-rand} 
\frac{1}{2} - \kappa \leq p(f(\textbf{x})|\textbf{u}) \leq \frac{1}{2} + \kappa
\ee
for some $0 < \kappa < \frac{1}{2}$. 
We say that Bell inequalities with property (\ref{eq:strong-rand}) guarantee \textit{strong randomness}. 

The Bell inequality we consider for the task of randomness amplification is a modified version of a Kochen-Specker game from \cite{Aolita}. The inequality involves two parties Alice and Bob, each making one of nine possible measurements and obtaining one of four possible outcomes, which is explained further in the Supplemental Material. Even though it does not guarantee the strong randomness in Eq.(\ref{eq:strong-rand}) for any function of the measurement outcomes $f(\textbf{x})$ for any input $\textbf{u}$, it has the redeeming feature of guaranteeing \textit{weak randomness} in the following sense. For all no-signaling boxes which algebraically violate the inequality, there exists one measurement setting $\textbf{u}^*$ and one outcome $\textbf{x}^*$ for this setting such that 
\ben
&&0 \leq p(\textbf{x} = \textbf{x}^* | \textbf{u} = \textbf{u}^*) \leq \gamma \nonumber \\
&& \forall \{p(\textbf{x}| \textbf{u})\} \; \; \; \text{s.t} \; \; \; \textbf{B} \cdot \{p(\textbf{x}| \textbf{u})\} = 0  
\een 
for some $0 < \gamma < 1$. The above fact is checked by linear programming and is shown in Lemma $1$ in the Supplemental Material. 

The novel technique in the form of a partial tomographic test, subsequent to the Bell test, allows us to extract randomness in this minimal scenario of weak randomness. This simply checks for the number of times a particular input-output pair $(\textbf{u}^*, \textbf{x}^*)$ appears, 
the analysis of this test is done by an application of the Azuma-Hoeffding inequality. We show that the SV source obeys a generalized Chernoff bound that ensures that with high probability when the inputs are chosen with such a source, the measurement setting $\textbf{u}^*$ appears in a linear fraction of the runs. Thus, conditioned on both tests in the protocol being passed (which happens with large probability with the use of the SV source and good quantum boxes by the honest parties), we obtain that with high probability over the input, the output is a source of linear min-entropy. 

This allows us to use known results on randomness extractors for two independent sources of linear min-entropy \cite{CG, one_bit_extr}, namely one given by the outputs of the measurement and the other given by the SV source. As shown in Proposition $16$ of \cite{our2}, one can use extractors secure against classical side information even in the scenario of general no-signaling adversaries by accepting a loss in the rate of the protocol, i.e., increasing the output error. The randomness extractor used in the protocol is a non-explicit extractor from \cite{CG}. Alternatively, there is an explicit extractor that can be employed in the protocol that has been found recently \cite{one_bit_extr}, but then it can produce just one bit of randomness. It also follows from \cite{our2} that there exists a protocol to obtain more bits with an explicit extractor using a device with three no-signaling components by employing additionally a de-Finetti theorem for no-signaling devices \cite{Brandao} 
(see Protocol II in \cite{our2}).

\textit{Conclusion and Open Questions.-}
We presented a device-independent protocol to amplify randomness in the minimal conditions under which such a task is possible, and used it to obtain secure random bits from an arbitrarily (but not fully) deterministic Santha-Vazirani source. The protocol uses a device consisting of only two non-signaling components, and works with correlations attainable by noisy quantum mechanical resources. Moreover, its correctness is not based on quantum mechanics and only requires the no-signaling principle. 

Important open questions still remain. 
One interesting question is whether the requirement of strict independence between the SV source and the devices can be relaxed to only require limited independence \cite{our3}. Another is to amplify the randomness of more general min-entropy sources that do not possess the structure of the Santha-Vazirani source. Finally, a significant open problem is to realize device-independent quantum key distribution with an imperfect source of randomness, tolerating a constant error rate and achieving a constant key rate.  

{\it Acknowledgments.}
The paper is supported by ERC AdG grant QOLAPS, EU grant RAQUEL and by Foundation for Polish Science TEAM project co-financed by the EU European Regional Development Fund. FB acknowledges support from EPSRC and Polish Ministry of Science and Higher Education Grant no.
IdP2011 000361. Part of this work was done in National Quantum Information Center of Gda\'{n}sk. 
Part of this work was done when F. B., R. R., K. H. and M. H. attended the program ``Mathematical Challenges in Quantum Information" at the Isaac Newton Institute for Mathematical Sciences in the University of Cambridge. 
 
\bibliographystyle{apsrev}


\textbf{Supplemental Material.}
Here, we give the formal proof of composable security for the device-independent protocol for randomness amplification using a device consisting of only two no-signaling components presented in the main text. 

Let us recall that the SV source is defined by the condition that bits $Y_i$ produced by the source obey
\begin{equation}
\label{SVdef-sup}
\frac{1}{2} - \varepsilon \leq P(Y_i = 0 | Y_{i-1}, \dots, Y_1) \leq \frac{1}{2} + \varepsilon
\end{equation}
for some $0 \leq \varepsilon < \frac{1}{2}$.
Let us also recall the notation from the main text. The honest parties and Eve share a no-signaling box $\{p(x, z|u', w)\}$ where $u' = \textbf{u'}_{\leq n} := (\textbf{u'}_{1}, \dots, \textbf{u'}_n)$ and $x = \textbf{x}_{\leq n} = (\textbf{x}_1, \dots, \textbf{x}_n)$ denote the input and output respectively of the honest parties for the $n$ runs of the protocol, with $w$ and $z$ the respective inputs and outputs of the adversary Eve. The devices held by the honest parties are separated into $m=2$ components with corresponding inputs and outputs $u'^i$ and $x^i$ respectively, for $i = 1, 2$, i.e., $u' = (u'^1, u'^2)$ and $x = (x^1, x^2)$. Here, $u'^i, x^i$ themselves denote the inputs and outputs of the $n$ runs of the protocol for party $i$, i.e., $u'^i = \textbf{u'}^i_{\leq n}$ and $x^i = \bx^i_{\leq n}$. 
Here, for the $j$-th run of the Bell test, the inputs of Alice are $\textbf{u'}^1_j$ and those of Bob are $\textbf{u'}^2_j$ with the corresponding outcomes $\textbf{x}^1_j$ and $\textbf{x}^2_j$ respectively, and the joint inputs of Alice and Bob in this run are $\textbf{u'}_j = (\textbf{u'}^1_j, \textbf{u'}^2_j)$ with corresponding joint outputs $\textbf{x}_j = (\textbf{x}^1_j, \textbf{x}^2_j)$. 
The honest parties draw bits $u$ from the SV source to input into the box, i.e., they set $u' = u$, they also draw a further $n$ bits $t$ which will be fed along with the outputs $x$ into the randomness extractor to obtain the output of the protocol $s := \text{Ext}(x,t)$. The adversary has classical information $e$ correlated to $u, t$. The box we consider for the protocol is given by the family of probability distributions $\{p(x,z,u,t,e|u',w)\}$

\section{Assumptions}
The Assumptions under which the Protocol is proven secure are formally stated in the main text. 
From Assumptions A1 and A2, as well as no-signaling  
\ben 
\label{eq:ns1-sup}
p(x|u', w) &=& p(x|u'), \nonumber \\
p(z|u',w) &=& p(z|w), \nonumber \\
p(x^i|u') &=& p(x^i|u'^i) \; \; \; i = 1, 2.
\een
and time-ordered-no-signaling assumptions,
\begin{eqnarray}
&&p(x^i_{k}|z,u'^i, w, u, t, e) = \nonumber \\
&& \; \; \;   p(x^i_{k} | z,u'^i_{\leq k},w,u, t, e) \; \; \; \forall k \in [n]
\label{eq:tons-sup}
\end{eqnarray}
we find that the distributions $\{p_w(x,z,u,t,e)\}$ satisfy (see \cite{our2}):
\ben
\label{eq:p-cond1}
&&p_w(x,u)=p(x,u) \\
\label{eq:p-cond2}
&& p_w (u,t,e)=p(u,t,e) \\
\label{eq:p-cond3}
&& \forall_w \hspace{0.1 cm} p_w(x,z|u,t,e)= p_w(x,z|u) \\
\label{eq:p-cond4}
&& \forall_w \hspace{0.1 cm} p_w(x,z|u,t,e)= p_w(x,z|u,e) \\
\label{eq:p-cond5}
&& p_w(x|z,u,t,e) = p_{z,t,e,w}(x|u) \; \text{is time-ordered no-signaling}  \nonumber \\
&&\\
\label{eq:p-cond6}
&& p_w(u|z,e) \; \text{,} \; p_w(t|z,u,e) \; \; \; \text{are SV sources} 
\een

The composable security criterion is given in terms of the distance $d_c$ defined as 
\be
\label{dist-comp-sup} 
d_c :=\sum_{s,e} \max_{w } \sum_{z} \left |p(s,z,e|w, \ACC ) - \frac{1}{|S|}p(z,e|w, \ACC )\right|.
\ee
Let us define the quantity $d'$ as
\ben 
\label{dprime}
d' &:=& \sum_{e} p(e|\ACC) \max_{w} \sum_{z,u} p(z,u|e,w,\ACC) \times \nonumber \\
&&\; \; \; \; \; \; \sum_{s} \left| p(s|z,w,u,e,\ACC) - \frac{1}{|S|} \right| 
\een
for any family of probability distributions $\{p(x,z,u,t,e|w)\}$.
Now, for each $e$, let $w_e$ and $p_{w_e}(x,z,u,t,e)$ denote the input of Eve and the corresponding probability distribution respectively that achieve the maximum $d'$ in Eq. (\ref{dprime}). By Assumption A1 and the no-signaling conditions, $p(e|w) = p(e)$ and $p(x,u|w) = p(x,u)$ so that the maximum is achieved by a distribution $q(x,z,u,t,e) = p(e) p_{w_e}(x,z,u,t|e)$. We can thus consider the quantity $d = d'$ given as
\be 
d = \sum_{z,u,e} q(z,u,e|\ACC) \sum_{s} \left| q(s|z,u,e,\ACC) - \frac{1}{|S|} \right|.
\ee
As shown in \cite{our2}, we have 
\be 
d_c \leq |S| d.
\ee
From the assumptions stated, it is seen that $q(x,u,z,t,e)$ obeys 
\ben
\label{q-cond} 
q(x,z|u,t,e) &=& q(x,z|u) \nonumber \\
q(x|z,u,t,e) &=& q_{t,e,z}(x|u) \; \text{is time-ordered no-signaling} \nonumber \\
q(u|z,e) &,& q(t|z,u,e) \; \text{obey the SV source conditions}. \nonumber \\
\een

\section{The Bell inequality} 
\label{sec: bellinequality}
The Bell inequality we consider for the task of randomness amplification is a modified version of the bipartite inequality in \cite{Aolita}. The inequality belongs to the class $(2, 9, 4)$ signifying that it involves two parties Alice and Bob, each making one of nine possible measurements and obtaining one of four possible outcomes. We label the measurement settings of Alice $\textbf{u}^1$ and those of Bob $\textbf{u}^2$ with $\textbf{u}^1, \textbf{u}^2 \in \{1, \dots, 9\}$. The corresponding outcomes of Alice are labeled $\textbf{x}^1$ and those of Bob $\textbf{x}^2$ with $\textbf{x}^1, \textbf{x}^2 \in \{1, \dots, 4\}$. Note that from the notation in the main text these inputs and outputs would correspond to a particular run of the protocol $\textbf{u}^i_j, \textbf{x}^i_j$. Acting on a box $\{P(\textbf{x} | \textbf{u})\}$ with $\textbf{x} = (\textbf{x}^1, \textbf{x}^2)$ and $\textbf{u} = (\textbf{u}^1, \textbf{u}^2)$, the Bell expression may be written as
\begin{equation}
\label{bip-Bell-ineq}
\textbf{B} \cdot \{P(\textbf{x}| \textbf{u})\} = \sum_{\textbf{x}, \textbf{u}} \textbf{B}(\textbf{x}, \textbf{u}) P(\textbf{x} | \textbf{u}) \geq 4,
\end{equation}
Here $\textbf{B}$ is an indicator vector with entries 
\begin{equation} \label{eq:bell-indicator}
\textbf{B}(\textbf{x}, \textbf{u}) = \left\{
     \begin{array}{lr}
       1 & : (\textbf{x}, \textbf{u}) \in S_B \\
       0 & : \text{otherwise} 
     \end{array}
   \right.
\end{equation} 
The minimum value achieved by local realistic theories for this combination of probabilities is $4$ while general no-signaling theories can achieve the algebraic minimum value of $0$. Crucially, there exist a quantum state and suitable measurements reaching this algebraic minimum. 

\begin{figure}
\scalebox{0.4}{\includegraphics{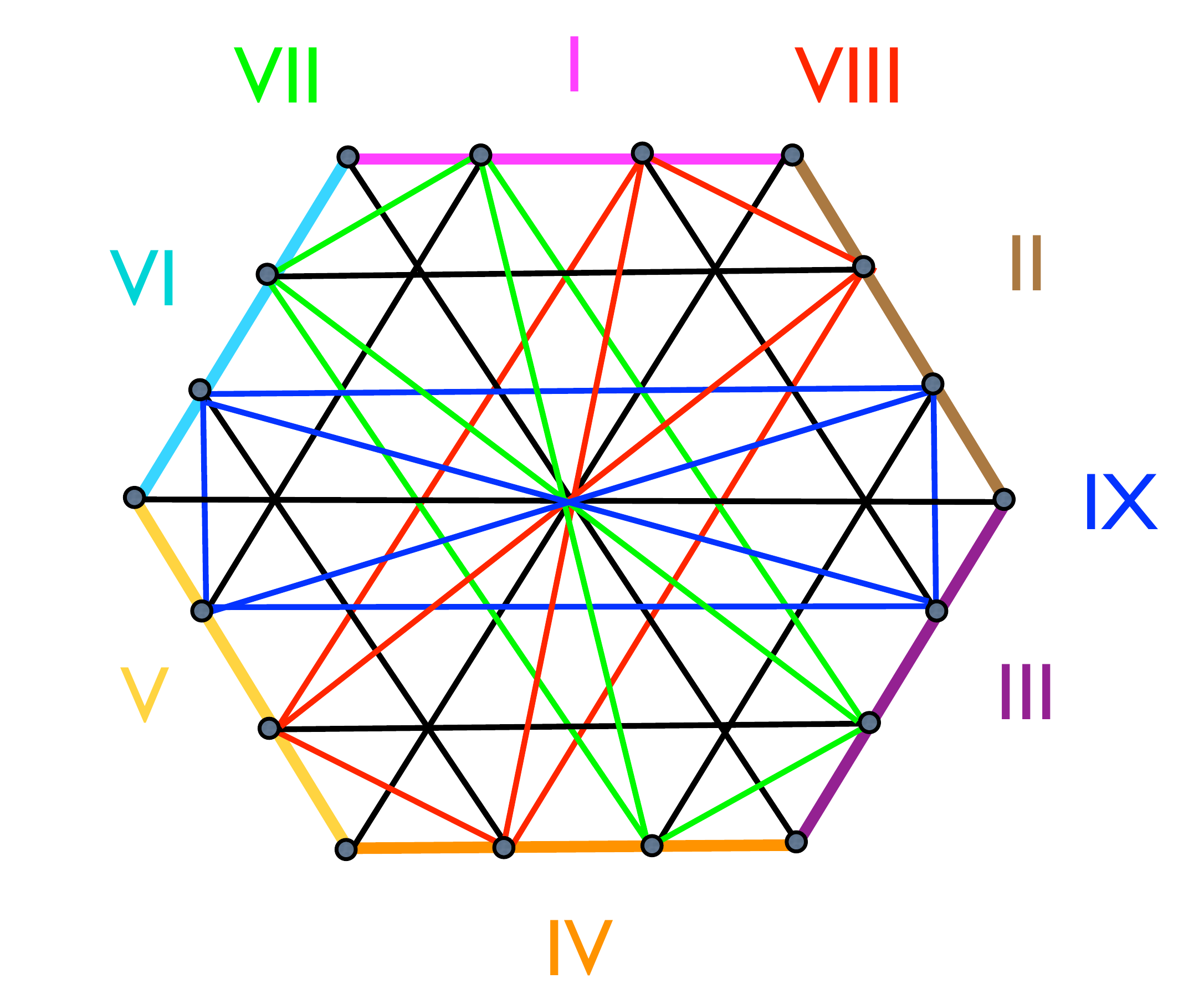}}
\caption{Illustration of the Kochen-Specker set used in formulating the bipartite Bell inequality}
\label{bip_Bell_fig}
\end{figure}

The set $S_{B} = \bigcup S_B^{\textbf{u}}$ for which $\textbf{B}(\textbf{x}, \textbf{u}) = 1$ is defined using the orthogonality hypergraph in Fig. \ref{bip_Bell_fig} which represents a Kochen-Specker set of vectors from \cite{Cabello} displaying state-independent contextuality in dimension $4$. In this graph, the nine measurements are represented by the nine colored hyperedges each giving four outcomes, where the vertices represent rank-one projectors corresponding to the outcomes. Each party performs the nine measurements corresponding to the KS set, 
the set $S_{B}$ consists of all $81$ pairs of measurements $\textbf{u}$. For each $\textbf{u}$, the pair of outcomes $\textbf{x} \in S_B^{\textbf{u}}$ if the vertex representing outcome $\textbf{x}^1$ in $\textbf{u}^1$ is connected by a hyperedge to the vertex representing outcome $\textbf{x}^2$ in $\textbf{u}^2$. 
A direct counting shows that out of the $4^2 \times 9^2 = 1296$ probabilities $P(\textbf{x} | \textbf{u})$, $504$ enter the Bell expression. Moreover, in any deterministic assignment of $1's$ and $0's$ to these probabilities respecting the no-signaling and normalization constraints, at least four probabilities are assigned value $1$ giving rise to the local realistic bound. In quantum theory and in general no-signaling theories however, all $504$ probabilities may be set to $0$ giving rise to the algebraic violation of the inequality.

In order to achieve the maximal violation within quantum theory, Alice and Bob share a maximally entangled state in dimension four, namely
\be
\label{eq:state}
| \Psi \rangle = \frac{1}{2} \sum_{i = 1}^{4} | i \rangle \otimes | i \rangle.
\ee
The measurements they each perform correspond exactly to the $18$ projectors defining the Kochen-Specker set in \cite{Cabello}. Specifically, these projectors correspond to the following vectors
\begin{widetext}
\be
\label{eq:measurements1}
\begin{tabular}{l l l l}
$| v_1\rangle = (1, 0, 0, 0)^T$ & $|v_2 \rangle = (0, 1, 0, 0)^T$ & $|v_3 \rangle = (0, 0, 1, 1)^T$ & $|v_4 \rangle = (0, 0, 1, -1)^T$  \\ $|v_5 \rangle = (1, -1, 0, 0)^T$ & $|v_6 \rangle =(1, 1, -1, -1)^T$ & $|v_7 \rangle= (1, 1, 1, 1)^T$ & $|v_8 \rangle = (1, -1, 1, -1)^T$ \\ $|v_ 9 \rangle = (1, 0, -1, 0)^T$ & $|v_{10} \rangle = (0, 1, 0, -1)^T$ & $| v_{11} \rangle = (1, 0, 1, 0)^T$ & $| v_{12} \rangle = (1, 1, -1, 1)^T$  \\ $|v_{13} \rangle = (-1, 1, 1, 1)^T$ & $|v_{14} \rangle = (1, 1, 1, -1)^T$ & $|v_{15} \rangle = (1, 0, 0, 1)^T$ & $|v_{16} \rangle = (0, 1, -1, 0)^T$  \\  $| v_{17} \rangle = (0, 1, 1, 0)^T$ & $| v_{18} \rangle = (0, 0, 0, 1)^T$
\end{tabular} 
\ee
\end{widetext}
The nine measurements are defined by the following nine bases 
\begin{widetext}
\be
\label{eq:measurements2}
\begin{tabular}{l l l}
$\textit{M}_1 = (|v_1 \rangle, |v_2 \rangle, |v_3 \rangle, |v_4 \rangle )$ & $\textit{M}_2 = (|v_4 \rangle, |v_5 \rangle, |v_6 \rangle, |v_7 \rangle ) $ & $\textit{M}_3 = (|v_7 \rangle, |v_8 \rangle, |v_9 \rangle, |v_{10} \rangle) $ \\ 
$\textit{M}_4 = ( | v_{10} \rangle, |v_{11} \rangle, |v_{12} \rangle, |v_{13} \rangle )$ & $ \textit{M}_5= (|v_{13} \rangle, |v_{14} \rangle, |v_{15} \rangle, |v_{16} \rangle )$ & $\textit{M}_6 = (|v_{16} \rangle, |v_{17} \rangle, |v_{18} \rangle, |v_{1} \rangle )$ \\ $ \textit{M}_7 = (|v_2 \rangle, |v_9 \rangle, |v_{11} \rangle, |v_{18} \rangle )$ & $\textit{M}_8 = (|v_3 \rangle, |v_5 \rangle, |v_{12} \rangle, |v_{14} \rangle )$ & 
$\textit{M}_9 = ( |v_6 \rangle, |v_8 \rangle, | v_{15} \rangle, |v_{17} \rangle )$
\end{tabular}
\ee
\end{widetext}
For this state and measurements all the probabilities entering the Bell expression are identically zero, so that algebraic violation is achieved. 

Apart from the fact that quantum mechanics violates the inequality, we would also like to ensure that a strong violation of the inequality guarantees randomness. Unfortunately, none of the bipartite Bell inequalities tested so far have this property. The above inequality though has the following redeeming feature. 
Let $\textbf{u}^* \equiv (1,2)$ be a particular pair of measurement settings and $\textbf{x}^* \equiv (1,3)$ a chosen pair of outcomes for this setting. For all no-signaling boxes which algebraically violate the inequality, it holds that 
\ben
&&0 \leq P(\textbf{x} = \textbf{x}^* | \textbf{u} = \textbf{u}^*) \leq \frac{3}{4} \nonumber \\
&& \forall \{P(\textbf{x}| \textbf{u})\} \; \; \; \text{s.t} \; \; \; \textbf{B} \cdot \{P(\textbf{x}| \textbf{u})\} = 0  
\een 
It should be noted that for the quantum box which algebraically violates the inequality defined by the above state and measurements, we have $P_q(\textbf{x} = \textbf{x}^* | \textbf{u} = \textbf{u}^*) = \frac{1}{16}$ so that upon maximal violation, we expect a fixed number of outputs $\textbf{x}^*$ for inputs $\textbf{u}^*$ in the experiment. Moreover, for boxes with a Bell value $\delta$, we will see in Lemma \ref{lem:B_SV_rand} that $0 \leq P(\textbf{x}^* | \textbf{u}^*) \leq \frac{1}{4}(3 + 2 \delta)$. So that, when one has large violation of the inequality and a sufficient number of outputs and inputs $(\textbf{x}^*, \textbf{u}^*)$, it must be the case that a sufficient number of runs in the experiment were done with boxes that yield randomness. 

\subsection{(Partial)Randomness from an observed Bell value} 
\label{subsec:SV-Bell}

Using the Azuma-Hoeffding inequality, we have that if the observed Bell value is small, then a linear fraction of the conditional boxes have a small Bell value for settings chosen with an SV source. To obtain a min-entropy source, we need to have that a linear fraction of the conditional boxes has randomness. In this section, we establish the consequence to randomness of the observed Bell value. 

Let $\textbf{U}$ denote all the settings appearing in the Bell expression. We consider first the uniform Bell value
\be
\label{uniform-Bell}
\overline{B}^U := \frac{1}{|\textsl{U}|} \textbf{B}. \{ P(\textbf{x} | \textbf{u}) \} = \frac{1}{|\textsl{U}|} \sum_{\bu,\bx}  B(\bx,\bu) P(\bx|\bu),
\ee
where $|\textsl{U}|$ denotes the cardinality of $\textbf{U}$, i.e. the total number of settings in the Bell expression ($|\textsl{U}| = 81$ for the Bell inequality we consider). If the Bell function $B(\bx,\bu)$ is properly chosen, one can prove using linear programming that if $\overline{B}^U$ is small, the probabilities of any output are bounded away from 1. However, since our inputs to each device are chosen using a SV source, we will be only able to estimate the value of the following expression
\be
\label{SV-Bell}
\overline{B}^{SV}= \sum_{\bu,\bx} \nu_{SV}(\bu)  B(\bx,\bu) P(\bx|\bu),
\ee
where $\nu_{SV}(\bu)$ is the distribution from an (unknown) SV source. Let us note that the number of bits needed by each party to choose their settings is $\left \lceil{\log{9}} \right \rceil = 4$, so that $\bu$ is chosen using $2 \left \lceil{\log{9}} \right \rceil = 8$ bits. We will show that for the Bell function, when $\overline{B}^{SV}$ is small, $\overline{B}^U$ is also small which implies randomness (for suitably chosen $\delta > 0$).

\begin{lemma}
\label{lem:B_SV_rand}
Consider a two-party no-signaling box $\{P (\textbf{x}| \textbf{u})\}$ satisfying 
\be
\overline{B}^{SV}\leq \delta,
\ee
for some constant $\delta \geq 0$, where $\overline{B}^{SV}$ is given by Eq. \eqref{SV-Bell} with $B(\bx,\bu)$ given by Eq. (\ref{eq:bell-indicator}). 
Then  for the particular measurement setting $\textbf{u}^*$ and particular output $\textbf{x}^*$, we have
\begin{equation} \label{SV-output-bound}
P( \textbf{x}= \textbf{x}^* | \textbf{u} = \textbf{u}^*) \leq   \frac14 \left(3+ \frac{2 \delta}{(\frac12 -\epsilon)^{8} }  \right).          
\end{equation}\end{lemma}

\begin{proof}
From the definition of an $\varepsilon$-SV source we have 
\begin{equation}
\left(\frac{1}{2} - \varepsilon \right)^{8} \leq \nu_{SV}(\bu) \leq \left(\frac{1}{2} + \varepsilon \right)^{8}.
\end{equation}
so that 
\begin{equation}
\label{SV-uniform}
\frac{1}{(\frac{1}{2} + \varepsilon)^{8} |\textsl{U}|} \overline{B}^{SV} \leq \overline{B}^{U} \leq \frac{1}{(\frac{1}{2} - \varepsilon)^{8} |\textsl{U}|} \overline{B}^{SV}
\end{equation}
%
%

We can therefore work with the Bell value for uniformly chosen settings, relating it to the Bell value with SV source settings through Eq. (\ref{SV-uniform}). For $\overline{B}^{SV}\leq \delta$, Eq.(\ref{SV-uniform}) gives that $\overline{B}^{U} \leq \frac{\delta}{(\frac{1}{2} - \varepsilon)^{8} |\textsl{U}|} =: \frac{\tilde{\delta}}{|\textsl{U}|}$. 

Consider a bipartite no-signaling box $P (\textbf{x}| \textbf{u})$ satisfying 
\begin{equation} 
\label{goodindividual}
\overline{B}^U := \frac{1}{|\textsl{U}|} \textbf{B}. \{ P(\textbf{x} | \textbf{u}) \} \leq \frac{\tilde{\delta}}{|\textsl{U}|},
\end{equation}
with $\textbf{B}$ the indicator vector for the Bell expression in Eq. \eqref{bip-Bell-ineq} and $|\textsl{U}| = 81$ the number of settings in the Bell expression. 

The maximum probability for the chosen output and input 
for the given (uniform) Bell value can be computed by the following linear program
\begin{eqnarray}
\label{lin-prog1}
&&\max_{ \{ P \}}: \textit{M}_{\textbf{u}^*, \textbf{x}^*}^T \cdot \{ P(\textbf{x} | \textbf{u}) \} \nonumber \\
&&s.t. \; \; \textit{A} \cdot \{ P( \textbf{x} | \textbf{u}) \} \leq \textit{c}.
\end{eqnarray}
Here, the indicator vector $\textit{M}_{\textbf{u}^*, \textbf{x}^*}$ is a $4^2 \times 9^2$ element vector with entries 
$M_{\textbf{u}^*, \textbf{x}^*}(\textbf{x}, \textbf{u}) = \texttt{I}_{\textbf{u} = \textbf{u}^*} \texttt{I}_{\textbf{x} = \textbf{x}^*}$, 
i.e., $M_{\textbf{u}^*, \textbf{x}^*}(\textbf{x}, \textbf{u}) = 1$ for $(\textbf{x}, \textbf{u}) = (\textbf{x}^*, \textbf{u}^*)$ and $0$ otherwise. 
The constraint on the box $\{P(\textbf{x} | \textbf{u})\}$ written as a vector with $4^2 \times 9^2$ entries is given by the matrix $\textit{A}$ and the vector $\textit{c}$. These encode the no-signaling constraints between the two parties, the normalization and the positivity constraints on the probabilities $P(\textbf{x} | \textbf{u})$. In addition, $\textit{A}$ and $\textit{c}$ also encode the condition that $\textbf{B}.\{ P(\textbf{x}| \textbf{u}) \} \leq \tilde{\delta}$ for a constant $\tilde{\delta} \geq 0$. \

The solution to the primal linear program in Eq. (\ref{lin-prog1}) can be bounded by a feasible solution to the dual program which is written as
\begin{eqnarray}
\label{dual-lin-prog1}
&&\min_{ \lambda_{\textbf{u}^*, \textbf{x}^*}}: \textit{c}^T \cdot \lambda_{\textbf{u}^*, \textbf{x}^*} \nonumber \\
&& s.t. \; \; \; \textit{A}^T \cdot \lambda_{\textbf{u}^*, \textbf{x}^*} = \textit{M}_{\textbf{u}^*, \textbf{x}^*}, \nonumber \\
&&\; \; \; \; \; \; \; \;  \lambda_{\textbf{u}^*, \textbf{x}^*} \geq 0.
\end{eqnarray}
We find a feasible $\lambda_{\textbf{u}^*, \textbf{x}^*}$ satisfying the constraints to the dual program above that gives $\textit{c}^T \lambda_{\textbf{u}^*, \textbf{x}^*} \leq \frac14 (3 + 2 \tilde{\delta} )$. \footnote{The explicit vector $\lambda_{\textbf{u}^*, \textbf{x}^*}$ that is feasible for the dual program in Eq. (\ref{dual-lin-prog1}) and gives the bound can be computed by standard techniques and is available upon request.} We therefore obtain by standard duality of linear programming that
\begin{equation}
P( \textbf{x} = \textbf{x}^*| \textbf{u} = \textbf{u}^* ) \leq  \frac14 (3+ 2 \tilde{\delta}).
\end{equation}
Noting that $\tilde{\delta} =\frac{\delta}{(\frac{1}{2} - \varepsilon)^{8}}$, we obtain the required bound.
\end{proof}

\section{From empirical values to true parameters of the box } 
\label{subsec:Azuma}
In this section, we state the lemmas based on the Azuma-Hoeffding inequality and the Generalized Chernoff bound which we will use to estimate the arithmetic average of Bell values for the conditional boxes  as well as the fraction of boxes which have a lower bound. 
Let us state the following Lemma \ref{lemmaazuma} based on the Azuma-Hoeffding inequality which we will use to estimate the arithmetic average of Bell values for the conditional boxes as well as the straightforward Lemma \ref{lem:linear_fraction} whose proofs can be found in \cite{our2}.
\begin{lemma} \label{lemmaazuma} 
Consider  arbitrary random variables $W_i$ for $i=0,1,\ldots,n$, and binary random variables $B_i$ for $i=1,\ldots n$ that are functions of $W_i$, i.e. $B_i=f_i(W_i)$ 
for some functions $f_i$. Let us denote $\overline{B}_i=\mathbb{E}(B_i|W_{i-1},\ldots,W_1,W_0)$ for $i=1,\ldots,n$ and (i.e. $\overline{B}_i$ are conditional means).
Define for $k = 1, \ldots, n$, the empirical average
\be
\Lazuma_k=\frac1k\sum_{i=1}^k B_i
\ee
and the arithmetic average of conditional means 
\be
\overline{\Lazuma}_k=\frac1k\sum_{i=1}^k \overline{B}_i.
\ee
Then we have 
\be
\Pr(|\Lazuma_n-\overline{\Lazuma}_n|\geq s)\leq 2 e^{-n\frac{s^2}{2}}
\label{eq:poor_man}
\ee
\end{lemma}

\begin{lemma}
\label{lem:linear_fraction}
If the arithmetic average $\overline{\Lazuma}_n$ of $n$ conditional means satisfies  $\overline{\Lazuma}_n\leq \delta$ for some parameter $\delta > 0$, then 
in at least $(1-\sqrt{\delta}) n$ of positions $i$ we have $\overline{B}_i\leq \sqrt{\delta}$
\end{lemma}

\subsection{Proving the lower bound for a fraction of boxes}
In this section, we estimate the fraction of boxes for which $q(\textbf{x}_i = \textbf{x}^*|\textbf{u}_i = \textbf{u}^*, \textbf{u}_{< i}, \textbf{x}_{< i}, z, e)$ is lower bounded by a constant.
To do so, we perform a test using the random variables $D^{u}_i(x)$ for any fixed $u$ 
\begin{displaymath}
   D^{u}_i(x) \defeq D(\textbf{x}_i, \textbf{u}_i) =  \left\{
     \begin{array}{lr}
       1 & : \textbf{x}_i = \textbf{x}^* \wedge \textbf{u}_i = \textbf{u}^* \\
       0 & : \text{otherwise} 
     \end{array}
   \right.
\end{displaymath} 
for $i=1,\ldots,n$.
The test function is defined as
\begin{equation}
\textit{S}_n(x,u) \defeq \frac{1}{n} \sum_{i=1}^{n} D(\textbf{x}_i, \textbf{u}_i) 
\end{equation}
with the corresponding average $\overline{\textit{S}}_n(x,u,z,e)$ defined as
\begin{equation}
\overline{\textit{S}}_n(x,u,z,e) \defeq \frac{1}{n} \sum_{i=1}^{n}  \mathbb{E}_{\sim q(\textbf{x}_i | \textbf{x}_{< i},  u, z, e)} D(\textbf{x}_i, \textbf{u}_i).
\end{equation}
The test checks if 
\begin{equation}
\label{eq:rand-test}
\textit{S}_n(x,u) \geq \mu_1
\end{equation} 
for a fixed $\mu_1 > 0$. 

We now show that when the test accepts, with probability $1 - 2 \exp \left( - n \frac{\mu_1^2}{8}  \right)$ at least $\frac{\mu_1 - 2\kappa}{2(1 - \kappa)} n$ boxes have randomness in the output for input setting $\textbf{u}^*$, specifically that $q(\textbf{x}_i = \textbf{x}^* | \textbf{u}_i = \textbf{u}^*, \textbf{u}_{<i}, \textbf{x}_{<i}, z, e) \geq \kappa$ for fixed $\kappa > 0$. 

\begin{lemma}
\label{lem:Azuma-count}
Assume that the test given by Eq. (\ref{eq:rand-test}) for the box $q(\textbf{x}_1, \ldots, \textbf{x}_n | \textbf{u}_1, \ldots, \textbf{u}_n, z, e)$ accepts (for fixed $\mu_1 > 0$). Consider the set $I_{\kappa}(u) \defeq \left\{ i : \textbf{u}_i = \textbf{u}^* \wedge q(\textbf{x}_i = \textbf{x}^* | \textbf{u}_i = \textbf{u}^*, \textbf{u}_{< i}, \textbf{x}_{<i}, z, e) \geq \kappa \right\}$. 
With probability at least $1 - 2 \exp \left( - n \frac{\mu_1^2}{8}  \right)$, $|I_{\kappa}(u)| \geq \frac{\mu_1 - 2\kappa}{2(1 - \kappa)} n$. 
\end{lemma}

\begin{proof}
When the test is passed, i.e., when $\textit{S}_n(x,u) \geq \mu_1$, by Lemma \ref{lemmaazuma} with probability at least $1 - 2 \exp \left( -n \frac{\mu_1^2}{8}  \right)$,
we have that $\overline{\textit{S}}_n(x,u,z,e) \geq \frac{\mu_1}{2}$. In other words, we have
\ben
\sum_{i} q(\textbf{x}_i = \textbf{x}^* | \textbf{u}_i = \textbf{u}^*, \textbf{u}_{< i}, \textbf{x}_{<i}, z, e) &\geq & \frac{\mu_1}{2},
\een
where we used the no-signaling condition $q(\textbf{x}_i = \textbf{x}^* | u, z, e) = q(\textbf{x}_i = \textbf{x}^* |  \textbf{u}_i = \textbf{u}^*, \textbf{u}_{< i}, \textbf{x}_{<i}, z, e)$. 
Consider the set $I_{\kappa}(u)$, we have that
\begin{equation}
(n - |I_{\kappa}(u)|) \kappa + |I_{\kappa}(u)| \geq \frac{\mu_1}{2} \; n
\end{equation}
or 
\begin{equation}
|I_{\kappa}(u)| \geq \frac{\mu_1 - 2\kappa}{2(1 - \kappa)} n.
\end{equation}
Therefore, with probability at least $1 - 2 \exp \left( - n \frac{\mu_1^2}{8}  \right)$ the set of boxes with $\textbf{u}_i = \textbf{u}^*$ and $q(\textbf{x}_i = \textbf{x}^* | \textbf{u}_i = \textbf{u}^*, \textbf{u}_{< i}, \textbf{x}_{<i}, z, e) \geq \kappa$ for fixed $\mu_1 > 0$, $0 < \kappa < \frac{1}{2}$ is of size at least $\frac{\mu_1 - 2\kappa}{2(1 - \kappa)} n$. \\
\end{proof}

\subsection{A min-entropy source from randomness of conditional boxes}
\label{subsec:min-entropy-source}

In this section we show that if a device is such that a linear number of conditional boxes have randomness (in the weak sense that the probability of the outputs is bounded away from one for any one setting and this particular setting appears a linear fraction of times), then the distribution on outputs constitutes a min-entropy source. 
Let any sequence $(\xuseq)$ be such that $\bx_i$ and $\bu_i$, $i\in \{1,\ldots,n\}$, are of the form of $\bx=(\bx^1,\bx^2)$ and $\bu=(\bu^1,\bu^2)$, respectively. Consider that with large probability over sequences $(\xuseq)$, a particular setting $\bu^*$ appears a linear fraction $\mu n$ times and that within this fraction, the probability of $\bx^*$ and its complementary outcome $\bar{\bx}^*$ is bounded away from $1$, then the total probability distribution is close in variational distance to a min-entropy source. 
To show this, we use the following lemma from \cite{our2}  
\begin{lemma}
\label{lem:min-entropy-0}
Fix any measure $P$ on the space of sequences $(\xuseq)$. 
Suppose that for a sequence $(\xuseq)$, there exists $\texttt{K} \subseteq [n]$ of size larger than $\mu n$, such that for all $l \in \texttt{K}$ we have $\bu_l = \bu^*$ and the conditional boxes $P_{\bx_{< l}, \bu_{< l}} (\bx_l | \bu_l, z, e) $ satisfy 
\be
P_{\bx_{< l}, \bu_{< l}} (\bx_l | \bu_l = \bu^*, z, e) \leq \gamma 
\label{eq:random-box2}.
\ee
Then, 
$\pxutot$ satisfies 
\begin{equation}
\pxutot\leq \gamma^{\mu n}.
\label{eq:hmin_cond}
\end{equation}
\end{lemma}

\section{Security Proof}
Let us first recall the definition of a min-entropy source and the notion of an independent source randomness extractor, specifying the extractor we will use to obtain randomness in our protocol. 
The min-entropy of a random variable $S$ is given by 
\begin{equation}
H_{\text{min}}(S) = \min_{s \in \text{supp}(S)} \log \frac{1}{P(S = s)},
\end{equation}
where $\text{supp}(S)$ denotes the support of $S$. For $S \in \{0,1\}^n$, the source is called an $(n, H_{\text{min}}(S))$ min-entropy source. An independent source extractor $\text{Ext}: (\{0,1\}^n)^k \rightarrow \{0,1\}^m$ is a function that acts on $k$ independent min-entropy sources and outputs $m$ bits that are $\xi$ close to uniform, i.e., for $k$ independent $(n, H_{\text{min}}(S_i))$ sources (with $i \in \{1,\dots, k\}$) we have
\begin{equation}
\Vert \text{Ext}(S_1, \ldots, S_k) - U_m \Vert_1 \leq \xi,
\end{equation}
where $\| . \|_1$ is the variational distance between the two distributions and $U_m$ denotes the uniform distribution on the $m$ bits. For use in Protocol I, we use a (non-explicit) deterministic extractor from \cite{CG} that, given two independent sources of min-entropy larger than $h$, outputs $\Omega(h)$ bits $2^{-\Omega(h)}$-close to uniform. Alternatively, in the protocol, one might also use the explicit extractor from \cite{one_bit_extr} that, given two independent sources of min-entropy at least $\log^{C}(h)$ for large enough constant $C$ outputs $1$ bit with error $h^{-\Omega(1)}$.

Let us define the set $Az^{\delta_{Az}}_{1}$ as
\begin{eqnarray}
\label{eq:Azuma1} 
Az^{\delta_{Az}}_{1} && \defeq \{ (z, u, e) : \nonumber \\
&& \Pr_{\sim q(x|z,u,e)} \left( \bar{\Lazuma}_n(x,u,z,e) \geq \Lazuma_n(x,u) + \delta_{Az} \right) \leq \epsilon_{Az1} \} \nonumber \\
\end{eqnarray}
and the cut 
\begin{eqnarray}
Az^{\delta_{Az}}_{1}(u) & \defeq & \{ (z, e) : (z,u,e) \in Az^{\delta_{Az}}_{1} \}. \nonumber \\
\end{eqnarray}
Let us also define the set $Az^{\mu_1}_{2}(u)$ for any fixed $u$ as
\ben 
\label{eq:Azuma2}
Az^{\mu_1}_{2}(u) && \defeq \{ (z,e) : \nonumber \\
&& \Pr_{\sim q(x|z,u,e)} \left( \overline{\textit{S}}_n(x,u,z,e) \leq \textit{S}_n(x,u) - \frac{\mu_1}{2} \right) \leq \epsilon_{Az2} \} \nonumber \\
\een
with $\epsilon_{Az1} = 2 e^{- n \frac{1}{4} \delta_{Az}^2}$ and $\epsilon_{Az2} = 2  e^{ - n \frac{\mu_1^2}{16} }$ and the set $Az(u)$ as
\be 
Az(u) \defeq Az^{\delta_{Az}}_{1}(u) \cap Az^{\mu_1}_{2}(u).
\ee
Note that despite the apparent similarity in the nomenclature of $Az^{\delta_{Az}}_{1}(u)$ and $Az^{\mu_1}_{2}(u)$, they differ in the respect that $Az^{\mu_1}_{2}(u)$ is a set of large measure for every $u$ (as seen in Eq. (\ref{eq:Azuma-measure})) while $Az^{\delta_{Az}}_{1}(u)$ is a set of large measure only for most (typical) $u$. 
Here 
\ben 
\Lazuma_n(x,u) &=& \frac{1}{n} \sum_{i=1}^{n} B(\bx_i, \bu_i), \nonumber \\
\bar{\Lazuma}_n(x,u,z,e) &=& \frac{1}{n} \sum_{i=1}^{n} \mathbb{E}_{q(\bx_i, \bu_i | \bx_{< i}, \bu_{< i}, z, e)} B(\bx_i, \bu_i).
\een 
Similarly, 
\ben 
\textit{S}_n(x,u) &=& \frac{1}{n} \sum_{i=1}^{n} D(\bx_i, \bu_i), \nonumber \\
\overline{\textit{S}}_n(x,u,z,e) &=& \frac{1}{n} \sum_{i=1}^{n} \mathbb{E}_{q(\bx_i, | \bx_{< i}, u, z, e)} D(\bx_i, \bu_i).
\een 
Applying Lemma \ref{lemmaazuma}, taking $W_0=(z,e)$, $W_i=(\textbf{x}_i, \textbf{u}_i)$ for $i=1,\ldots, n$, we obtain by a direct application of the Markov inequality that
\ben 
\label{eq:Azuma-measure}
&&\sum_{(z,u,e) \in Az^{\delta_{Az}}_{1}} q(z,u,e) \geq 1 - \epsilon_{Az1} \nonumber \\
&&\sum_{(z,e) \in Az^{\mu_1}_{2}(u)} q(z,e|u) \geq 1 - \epsilon_{Az2}.
\een 
To elaborate, we get from Lemma \ref{lemmaazuma} that
\ben
&&\Pr_{(x,u,z,e) \sim q(x,u,z,e)} \left(\bar{\Lazuma}_n(x,u,z,e) \geq \Lazuma_n(x,u) + \delta_{Az} \right) \leq  \epsilon_{Az1}^2 \nonumber \\
&&\Pr_{(z,u,e) \sim q(z,u,e)} [\Pr_{x \sim q(x|z,u,e)} \left(\bar{\Lazuma}_n(x,u,z,e) \geq \Lazuma_n(x,u) + \delta_{Az} \right) \nonumber \\ 
&& \qquad \qquad \qquad \qquad \qquad \qquad \qquad \geq  \epsilon_{Az1} ] \leq \epsilon_{Az1}
\een
and the second inequality in Eq.(\ref{eq:Azuma-measure}) is obtained similarly. 
  
Also, as stated previously we define the sets $\ACC_1$ and $\ACC_2$ as the sets of $(x,u)$ for which the tests in the protocol are passed, i.e., 
\ben 
\ACC_1 &\defeq & \{(x,u) : \Lazuma_n(x,u) \leq \delta\} \nonumber \\
\ACC_2 &\defeq & \left\lbrace (x,u) : \textit{S}_n(x,u) \geq \mu_1 \right\rbrace,
\een
and the set $\ACC = \ACC_1 \cap \ACC_2$ of $(x,u)$ for which both tests in the protocol are passed. Let us also define 
\be 
\ACC_u \defeq \{ x : (x,u) \in \ACC \}.
\ee

We are now ready to formulate the following lemma.
\begin{lemma}
\label{lem:min-ent}
Consider the measure $q(x,z,u,t,e)$ satisfying Eq.(\ref{q-cond}). For constant $\delta_1 > 0$, we have that
\ben 
&&\Pr_{\sim q(z,u,e|\ACC)} \left( \max_{x} q(x|z,u,e,\ACC) \leq \sqrt{\frac{\delta_1}{q(\ACC)}} \right)  \nonumber \\
&& \; \; \; \; \; \geq 1 - \sqrt{\frac{\delta_1}{q(\ACC)}}.
\een
\end{lemma}

\begin{proof}

Let us write
\ben 
\label{eq:ustar-split}
&& \sum_{z,u,e} q(z,u,e|\ACC) \max_{x} q(x|z,u,e,\ACC) \nonumber \\
&= & \sum_{(z,u,e) \notin Az^{\delta_{Az}}_{1}} q(z,u,e|\ACC) \max_{x} q(x|z,u,e,\ACC)  \nonumber \\
&& + \sum_{(z,u,e) \in Az^{\delta_{Az}}_{1}} q(z,u,e|\ACC) \max_{x} q(x|z,u,e,\ACC). \nonumber \\
\een
and bound the two terms separately. 
The first term can be simply bounded as 
\ben 
\label{eq:ustar-split-1}
&& \sum_{(z,u,e) \notin Az^{\delta_{Az}}_{1}} q(z,u,e|\ACC) \max_{x} q(x|z,u,e,\ACC) \nonumber \\
&& \stackrel{\max_{x} q(x|z,u,e,\ACC) \leq 1}{\leq}  \sum_{(z,u,e) \notin Az^{\delta_{Az}}_{1}} q(z,u,e|\ACC) \nonumber \\
&& \stackrel{q(z,u,e,\ACC) \leq q(z,u,e)}{\leq} \sum_{(z,u,e) \notin Az^{\delta_{Az}}_{1}}  \frac{q(z,u,e)}{q(\ACC)} \nonumber \\
&& \stackrel{Eq.(\ref{eq:Azuma-measure})}{\leq} \sum_{u} \frac{\epsilon_{Az1}}{q(\ACC)}.
\een

For the second term, with $(z,u,e) \in Az^{\delta_{Az}}_{1}$, we have that for fixed $u$, $(z,e) \in Az^{\delta_{Az}}_{1}(u)$. We therefore split the second term as
\begin{widetext}
\ben 
\label{eq:ustar-split-2}
&& \sum_{(z,u,e) \in Az^{\delta_{Az}}_{1}} q(z,u,e|\ACC) \max_{x} q(x|z,u,e,\ACC) \nonumber \\
&=& \sum_{\substack{u \\ (z,e) \in Az^{\delta_{Az}}_{1}(u) \cap Az^{\mu_1}_{2}(u)}} q(z,u,e|\ACC) \max_{x} q(x|z,u,e,\ACC) 
+\sum_{\substack{u \\ (z,e) \in Az^{\delta_{Az}}_{1}(u) \cap \left(Az^{\mu_1}_{2}(u) \right)^c}} q(z,u,e|\ACC) \max_{x} q(x|z,u,e,\ACC), \nonumber \\
\een
\end{widetext}
where $\left(Az^{\mu_1}_{2}(u) \right)^c$ denotes the complement of the set $Az^{\mu_1}_{2}(u)$. Let us first consider the case when $(z,e) \in  Az^{\delta_{Az}}_{1}(u) \cap Az^{\mu_1}_{2}(u)$, i.e., $(z,e) \in Az(u)$. We define the sets
\ben 
X^{(z,u,e)}_{g1} &=& \{x :  \bar{\Lazuma}_{n}(x, u, z, e) \leq \Lazuma_{n}(x,u) + \delta_{Az} \}, \nonumber  \\
X^{(z,u,e)}_{g2} &=& \{x : \overline{\textit{S}}_n(x, u, z, e) \geq \textit{S}_n(x,u) - \frac{\mu_1}{2} \}, \nonumber \\
\een
and the complements $\left(X^{(z,u,e)}_{g1}\right)^{c}, \left(X^{(z,u,e)}_{g2}\right)^{c}$. 

By the definition of $Az^{\delta_{Az}}_{1}(u)$, for $(z,e) \in Az^{\delta_{Az}}_{1}(u)$ and $x \in \left(X^{(z,u,e)}_{g1}\right)^{c}$, we have 
\be 
q(x|z,u,e) \leq \epsilon_{Az1}
\ee
for $\epsilon_{Az1} = 2 e^{- n \frac{1}{4} \delta_{Az}^2}$.
Similarly, by the definition of $Az^{\mu_1}_{2}(u)$, for $(z,e) \in Az^{\mu_1}_{2}(u)$ and $x \in \left(X^{(z,u,e)}_{g2}\right)^{c}$, we have 
\be 
q(x|z,u,e) \leq \epsilon_{Az2}
\ee 
for $\epsilon_{Az2} = 2  e^{ - n \frac{\mu_1^2}{16} }$.
Therefore, for $(z,e) \in Az(u)$ and $x \in \left(X^{(z,u,e)}_{g1} \cap  X^{(z,u,e)}_{g2} \right)^{c} \cap \ACC_u$, we have that
\be 
\label{eq:x-bound-acc1}
q(x|z,u,e) \leq \epsilon_{Az1} + \epsilon_{Az2}.
\ee

Now let us look at the case when $(z,e) \in  Az(u)$ and $x \in \left(X^{(z,u,e)}_{g1} \cap  X^{(z,u,e)}_{g2} \right) \cap \ACC_u$. By the definition of $\ACC_1$, we have $\Lazuma_n(x,u) \leq \delta$, and by the definition of $X^{(z,u,e)}_{g1}$ we have that 
\be 
\bar{\Lazuma}_n(x,u,z,e) \leq \delta + \delta_{Az}.
\ee
By Lemma \ref{lem:linear_fraction}, for at least $\mu_2 n$ positions $i$ where $\mu_2 = 1 -\sqrt{\delta + \delta_{Az}}$, there is
\be 
\mathbb{E}_{q(\bx_i, \bu_i | \bx_{< i}, \bu_{< i}, z, e)} B(\bx_i, \bu_i) \leq \sqrt{\delta + \delta_{Az}} = \sqrt{2 \delta},
\ee
where we have simply set $\delta_{Az} = \delta$ for constant $\delta > 0$.
Therefore, by Lemma \ref{lem:B_SV_rand}, at these $\mu_2 n$ positions $i$, we have that for the particular input and output pair $\bu_i = \bu^*$ and $\bx_i = \bx^*$
\be 
q_{x_{<i}, u_{< i}, z, e}( \textbf{x}_i  = \textbf{x}^*| \textbf{u}_i = \textbf{u}^*) \leq   \frac14 \left(3+ \frac{2 \sqrt{2\delta}}{(\frac12 -\epsilon)^{8} }  \right).
\ee 
Note that we will choose $\delta$ such that
\ben 
\label{eq:delta-bound1}
&&\frac14 \left(3+ \frac{2 \sqrt{2\delta}}{(\frac12 -\epsilon)^{8} }  \right) < 1 \nonumber \\
&\text{i.e.,}\; \; & 0 < \delta < \frac{(\frac{1}{2} - \varepsilon)^{16}}{8}
\een
to have the above probability bounded below unity. 
Similarly, by the definition of $\ACC_2$, $\textit{S}_n(x,u) \geq \mu_1$, and by the definition of  $X^{(z,u,e)}_{g2}$, we have that
\be 
\overline{\textit{S}}_n(x,u,z,e) \geq \frac{\mu_1}{2}.
\ee
By Lemma \ref{lem:Azuma-count}, for at least $\mu_3 n$ positions $i$, where $\mu_3 = \frac{\mu_1 - 2\kappa}{2(1 - \kappa)}$ for fixed $\kappa > 0$, we have
\be 
q_{\textbf{x}_{<i}, \textbf{u}_{<i}, z, e}( \textbf{x}_i = \textbf{x}^*| \textbf{u}_i = \textbf{u}^*) \geq \kappa.
\ee
Therefore, for $(z,e) \in Az(u)$ and $x \in \left(X^{(z,u,e)}_{g1} \cap  X^{(z,u,e)}_{g2} \right) \cap \ACC_u$, we have that there are at least $\mu_4 n$ positions $i$ with $\mu_4 = (\mu_3 + \mu_2 - 1)$ for which 
\be 
\label{eq:x-bound-acc2}
q_{\textbf{x}_{<i}, \textbf{u}_{<i}, z, e}(\textbf{x}_i | \textbf{u}_i = \textbf{u}^*) \leq \gamma
\ee
for $\textbf{x}_i = \textbf{x}^*$ as well as $\textbf{x}_i \neq \textbf{x}^*$. Here, 
\be 
\gamma =  \max{ \left\lbrace \left(1 - \kappa \right), \frac14 \left(3+ \frac{2 \sqrt{2\delta}}{(\frac12 -\epsilon)^{8} }  \right) \right\rbrace }.
\ee
In order to have $\mu_4 > 0$, i.e., $\mu_3 + \mu_2 > 1$ we will choose constant $\delta > 0$ such that 
\ben
\label{eq:delta-bound2}
&&\frac{\mu_1 - 2 \kappa}{2(1-\kappa)} - \sqrt{2 \delta} > 0, \nonumber \\
&\text{i.e.,}& \delta < \frac{1}{2} \left[ \frac{\mu_1 - 2\kappa}{2(1-\kappa)} \right]^2.
\een
Combining Eq. (\ref{eq:delta-bound1}) and Eq.(\ref{eq:delta-bound2}) we have that 
\be 
\delta < \min \left\lbrace \frac{(\frac{1}{2} - \varepsilon)^{16}}{8}, \frac{1}{2} \left[ \frac{\mu_1 - 2\kappa}{2(1-\kappa)} \right]^2 \right\rbrace
\ee
Therefore, for any $(z,e) \in Az(u)$ and $x \in \ACC_u$, combining Eq. (\ref{eq:x-bound-acc1}) and Eq.(\ref{eq:x-bound-acc2}) we have from Lemma \ref{lem:min-entropy-0} that
\ben 
\label{eq:x-bound-acc}
\max_{x} q(x|z,u,e,\ACC) &=& \frac{\max_{x \in \ACC_u} q(x|z,u,e)}{q(\ACC|z,u,e)} \nonumber \\
&\leq & \frac{\max\{ \epsilon_{Az1} + \epsilon_{Az2}, \gamma^{ \mu_4 n} \}}{q(\ACC|z,u,e)}. \nonumber \\
\een
From the above considerations, we can bound 
\ben
\label{eq:ustar-split-2-1}
&&\sum_{\substack{u \\ (z,e) \in  Az(u)}} q(z,u,e|\ACC) \max_{x} q(x|z,u,e,\ACC) \nonumber \\
&& \; \; \; \stackrel{\text{Eq. (\ref{eq:x-bound-acc})}}{\leq} \sum_{\substack{u \\ (z,e) \in Az(u)}} q(z,u,e|\ACC)  \frac{\max\{ \epsilon_{Az1} + \epsilon_{Az2}, \gamma^{\mu_4 n} \}}{q(\ACC|z,u,e)} \nonumber \\ 
&& \; \; \; \leq \max\{ \epsilon_{Az1} + \epsilon_{Az2}, \gamma^{\mu_4 n} \} \sum_{(z,u,e)} \frac{q(z,u,e)}{q(\ACC)} \nonumber \\
&& \; \; \;  \leq \frac{\max\{ \epsilon_{Az1} + \epsilon_{Az2}, \gamma^{\mu_4 n} \}}{q(\ACC)}.
\een
We can also simply bound
\begin{widetext}
\ben 
\label{eq:ustar-split-2-2}
&&\sum_{\substack{u \\ (z,e) \in Az^{\delta_{Az}}_{1}(u) \cap \left(Az^{\mu_1}_{2}(u) \right)^c}}  q(z,u,e|\ACC) \max_{x} q(x|z,u,e,\ACC) \nonumber \\
&\leq & \sum_{\substack{u \\ (z,e) \in Az^{\delta_{Az}}_{1}(u) \cap \left(Az^{\mu_1}_{2}(u) \right)^c}} q(z,u,e|\ACC) \stackrel{q(z,u,e,\ACC) \leq q(z,u,e)}{\leq}  \sum_{\substack{u \\ (z,e) \in \left(Az^{\mu_1}_{2}(u) \right)^c}} \frac{q(u)q(z,e|u)}{q(\ACC)} 
\stackrel{Eq.(\ref{eq:Azuma-measure})}{\leq}  \frac{\epsilon_{Az2}}{q(\ACC)}.
\een
\end{widetext}
Inserting the bounds from Eqs. (\ref{eq:ustar-split-2-1}) and (\ref{eq:ustar-split-2-2}) into Eq.(\ref{eq:ustar-split-2}) gives
\ben 
\label{eq:ustar-split-2f}
&& \sum_{(z,u,e) \in Az^{\delta_{Az}}_{1}} q(z,u,e|\ACC) \max_{x} q(x|z,u,e,\ACC) \nonumber \\
&&  \leq \frac{\epsilon_{Az1} + 2 \epsilon_{Az2} +\gamma^{\mu_4 n} }{q(\ACC)}
\een
Finally, inserting the bounds from Eqs.(\ref{eq:ustar-split-1}) and (\ref{eq:ustar-split-2f}) into Eq. (\ref{eq:ustar-split}) gives
\ben 
&& \sum_{(z,u,e)} q(z,u,e|\ACC) \max_{x} q(x|z,u,e,\ACC) \nonumber \\
&& \leq \frac{2(\epsilon_{Az1} + \epsilon_{Az2}) +\gamma^{\mu_4 n} }{q(\ACC)}
\een

Applying Markov inequality, setting $\delta_1 = 2(\epsilon_{Az1} + \epsilon_{Az2})  + \gamma^{\mu_4 n}$, we get that
\ben 
&&\Pr_{\sim q(z,u,e|\ACC)} \left( \max_{x} q(x|z,u,e,\ACC) \leq \sqrt{\frac{\delta_1}{q(\ACC)}} \right) \nonumber \\
&& \; \; \; \; \;  \geq 1 - \sqrt{\frac{\delta_1}{q(\ACC)}}.
\een
This completes the proof.
\end{proof}
We now note the following lemma which follows from the assumptions stated in the text (for a proof see \cite{our2})
\begin{lemma}
\label{lem:Markov_acc}
For any probability distribution $q(x,z,u,t,e)$ satisfying Eq.(\ref{q-cond}) it holds that
\be 
q(x|z,u,t,e,ACC) = q(x|z,u,ACC).
\ee
\end{lemma} 
We use Lemma \ref{lem:Markov_acc} along with Lemma \ref{lem:min-ent} to obtain the following theorem whose proof follows a similar statement in \cite{our2} showing that either the tests in the protocol are passed with vanishing probability or we obtain $\Omega(n^{1/4})$ ($|S| = 2^{\Omega(n^{1/4})}$) secure random bits.  
\begin{thm}
\label{thm:main_protocol_I}
Let $n$ denote the number of runs in Protocol I and suppose we are given $\epsilon>0$. For fixed $\mu_1 > 0$, $0 < \kappa < \frac{\mu_1}{2}$, set $\delta > 0$ such that  
\be
\delta < \min \left\lbrace \frac{(\frac{1}{2} - \varepsilon)^{16}}{8}, \frac{1}{2} \left[ \frac{\mu_1 - 2\kappa}{2(1-\kappa)} \right]^2. \right\rbrace
\ee
Then for any probability distribution $p_w(x,z,u,t,e)$  
satisfying Eqs.  \eqref{eq:p-cond1}-\eqref{eq:p-cond6}
there exists a {non-explicit} extractor $s(x,t)$ with $|S| = 2^{\Omega(n^{1/4})}$ values, such that 
\be
\dcintro \cdot p(\ACCd) \leq 2^{-\Omega(n^{1/4})}, 
\ee
where $\dcintro$ is given by \eqref{dist-comp-sup} as 
\be
d_c :=\sum_{s,e} \max_{w } \sum_{z} \left |p(s,z,e|w, \ACC ) - \frac{1}{|S|}p(z,e|w, \ACC )\right|.
\ee
{Alternatively, one can use an explicit extractor $s'(x,t)$ producing a single bit of randomness with 
\be
\dcintro \cdot p(\ACCd) \leq 2^{-\Omega(n^{1/(2C)})}, 
\ee
for some constant $C$.} 
\end{thm}

\section{Passing the tests with quantum boxes}
Finally, we check that for suitable parameters $\delta$ and $\mu_1$ both tests in the protocol are passed with the use of good quantum boxes by the honest parties. 
\subsection{Generalized Chernoff bound for Santha-Vazirani sources}
The final part of the proof is to show that if the honest parties use good quantum boxes, the tests in the protocol are passed with high probability. We first show that the Santha-Vazirani source satisfies an exponential concentration property given by the following generalized Chernoff bound, which will imply that the second test in the protocol is feasible, i.e., that in a linear fraction of the runs the setting $\textbf{u}^*$ appears. 
\begin{thm}(\textbf{Generalized Chernoff bound})\cite{Panconesi-Srinivasan,Impagliazzo-Kabentes}
\label{thm:chernoff}
Let $X_i$ for $i\in[n]$ be Boolean random variables such that for some $0\leq \zeta \leq 1$, we have that, for every subset $S \subseteq [n]$
$\Pr\left[ \wedge_{i\in S}X_i =1 \right] \leq \zeta^{|S|}$. Then, for any $0\leq \zeta \leq \gamma \leq 1$ 
\be
\Pr \left[ \sum_{i=1}^n X_i \geq \gamma n \right] \leq  e^{-nD(\gamma||\zeta)},
\ee
where $D( \cdot || \cdot )$ is the relative entropy function. In particular $D(\gamma||\zeta) \geq 2(\gamma - \zeta)^2$.
\end{thm}

We show now that the SV source satisfies the assumption of the above theorem, i.e., that probability of not obtaining the input 
$\textbf{u}^*$ in a subset of size $k$ is upper bounded by $\zeta^k$ for $\zeta =  \left[1 - \left({1\over 2}-\epsilon\right)^{2 m} \right]$ with $2m$ being the number of bits the two parties need to choose a single $\textbf{u}$ ($2m = 2 \left \lceil{\log{9}} \right \rceil = 8$ for the Bell inequality we consider). 
%
\begin{lemma}
\label{SV-chernoff} 
For any non-empty subset of $k$ indices $(i_1,...,i_k) \subseteq [n]$, and $n$ consecutive instances of random variable $U$ chosen according to measure $\nu$ using $2 m n$ bits from an $\epsilon$-SV source (where $2m$ is the number of bits required to choose a single instance $\bu$), for any fixed $\textbf{u}^*$ in the range of $U$, we have
\be
Pr_{\sim \nu}(\textbf{u}_{i_1} \neq \textbf{u}^*,\dots, \textbf{u}_{i_k} \neq \textbf{u}^*) \leq \left[1 - \left({1\over 2}-\epsilon\right)^{2 m} \right]^{k}
\ee
\end{lemma}
\begin{proof}
Let us assume, w.l.o.g. that $i_k\geq i_{k-1}\geq ... \geq i_1$. We have
\begin{widetext}
\begin{eqnarray}
&&\Pr_{\sim \nu}(\textbf{u}_{i_1} \neq \textbf{u}^*,\dots, \textbf{u}_{i_k} \neq \textbf{u}^*) \nonumber \\
&&=\sum_{\{\textbf{u}_{i_j}\}: {i_j} \notin \{i_1,\dots,i_k\}} \Pr_{\sim \nu}(\textbf{u}_1, \dots, \textbf{u}_{i_1} \neq \textbf{u}^*, \dots, \textbf{u}_{i_k} \neq \textbf{u}^*, \dots, \textbf{u}_n) \nonumber \\
&&=\sum_{ \{\textbf{u}_{i_j}\}: {i_j} \notin \{i_1,\dots,i_k\}} \Pr_{\sim \nu}(\textbf{u}_1) \Pr_{\sim \nu}(\textbf{u}_{i_1} \neq \textbf{u}^*|\textbf{u}_1,\dots, \textbf{u}_{i_1-1}) \dots \Pr_{\sim \nu}(\textbf{u}_{i_k} \neq \textbf{u}^*| \textbf{u}_1,\dots, \textbf{u}_{i_k -1}) \dots \Pr_{\sim \nu}(\textbf{u}_n | \textbf{u}_1, \dots, \textbf{u}_{n-1}) \nonumber \\
&& \leq \left[1 - \left({1\over 2}-\epsilon\right)^{2 m} \right]^{k} 
\end{eqnarray}
\end{widetext}
The last inequality is obtained by noting that for terms with $i_j \in \{i_1, \dots, i_k \}$, by the definition of the SV source $P(\textbf{u}_{i_j} \neq \textbf{u}^*| \textbf{u}_1, \dots, \textbf{u}_{i_j - 1}) \leq \left[1 - \left({1\over 2}-\epsilon\right)^{2 m} \right]$ with $2m$ being the number of bits required to obtain any input $\textbf{u}$, and for the terms with $i_j \notin \{i_1, \dots, i_k\}$, the sum over $\textbf{u}_{i_j}$ gives unity by normalization. 
\end{proof}
%
Consider the random variable $X_i$ defined as
\begin{displaymath}
X_i \defeq  \left\{
     \begin{array}{lr}
       1 & : \textbf{u}_i \neq \textbf{u}^* \\
       0 & : \text{otherwise} 
     \end{array}
   \right.
\end{displaymath} 
for $\textbf{u}_i$ chosen using the SV source $\nu(\cdot)$. Theorem \ref{thm:chernoff} together with Lemma \ref{SV-chernoff} gives that 
\be
\Pr \left[ \sum_{i=1}^n X_i \geq \gamma n \right] \leq  e^{-2n(\gamma - \zeta)^2},
\ee
or equivalently
\be 
\label{eq:chernoff-bound-1}
\Pr \left[\sum_{i=1}^n X_i < \gamma n \right] \geq 1 - e^{-2n(\gamma - \zeta)^2},
\ee
for $\zeta = \left[1 - \left({1\over 2}-\epsilon\right)^{2 m} \right]$ and $0\leq \zeta \leq \gamma \leq 1$. For $\tilde{U}(u) := \{ i: \textbf{u}_i = \textbf{u}^* \}$ and $Ch \defeq \{ u: |\tilde{U}(u)| \geq \mu_5 n \}$ for some constant $\mu_5 > 0$, Eq. (\ref{eq:chernoff-bound-1}) gives that 
\be 
\label{eq:chernoff-bound-2}
\sum_{u \in Ch} \nu(u) \geq 1 - e^{-2n(1 - \mu_5 - \zeta)^2}.
\ee
Therefore, we obtain that with probability $1 -  e^{-2n(1 - \mu_5 - \zeta)^2}$, $\textbf{u}_i = \textbf{u}^*$ for a fraction $\mu_5$ of the $n$ runs. We note that with the use of the state and measurements from Eqs.(\ref{eq:state}), (\ref{eq:measurements1}) and (\ref{eq:measurements2}), we obtain a box  $\{P_q(\textbf{x}| \textbf{u})\}$ that achieves maximal violation of the Bell inequality, i.e., $\textbf{B}. \{P_q(\textbf{x} | \textbf{u})\} = 0$ and also has $P_q(\textbf{x} = \textbf{x}^* | \textbf{u} = \textbf{u}^*) = \frac{1}{16}$. Therefore, for suitably chosen $\delta, \mu_1 > 0$ the two tests in the protocol are passed with high probability with the use of good quantum boxes.

\bibliographystyle{apsrev}


\end{document}